\documentclass[12pt,a4paper]{article}

\usepackage{amssymb}
\usepackage{amsfonts}
\usepackage{graphicx}
\usepackage{amsmath}
\usepackage{latexsym}
\usepackage{amssymb}
\usepackage{pstricks}
\usepackage{mathrsfs}
\usepackage{fancyhdr}
\usepackage{amsthm}

 \newtheorem{thm}{Theorem}
 \newtheorem{cor}[thm]{Corollary}
 \newtheorem{prop}[thm]{Proposition}

\newtheorem{theorem}{Theorem}
\newtheorem{lem}[theorem]{Lemma}
\newtheorem{rem}[theorem]{Remark}

\newcommand{\om}{\Omega}
\newcommand{\pom}{\partial\Omega}

\newcommand{\iom}{\int_\om}

\newcommand{\bfC}{\mbox{\boldmath{$C$}}}

\newcommand{\bfx}{\mbox{\boldmath{$x$}}}
\newcommand{\bfn}{\mbox{\boldmath{$n$}}}
\newcommand{\bff}{\mbox{\boldmath{$f$}}}
\newcommand{\bfF}{\mbox{\boldmath{$F$}}}

\newcommand{\bfg}{\mbox{\boldmath{$g$}}}

\newcommand{\bfz}{\mbox{\boldmath{$z$}}}

\newcommand{\bfu}{\mbox{\boldmath{$u$}}}
\newcommand{\bfw}{\mbox{\boldmath{$w$}}}
\newcommand{\bfv}{\mbox{\boldmath{$v$}}}

\newcommand{\bfe}{\mbox{\boldmath{$e$}}}

\newcommand{\bfalpha}{\mbox{\boldmath{$\alpha$}}}

\newcommand{\bfdelta}{\mbox{\boldmath{$\delta$}}}

\begin{document}

%
%
%
%
%
%
%
%
%

\title
 {A Note on Regularity and Uniqueness of Natural
 Convection with Effects of Viscous
 Dissipation in 3D Open Channels}

\author{Michal Bene\v{s}}
\date{\small Department of Mathematics,
\\
Faculty of Civil Engineering,
\\
Czech Technical University in Prague,
\\
Th\'{a}kurova 7, 166 29 Prague 6, Czech Republic
\\
email: benes@mat.fsv.cvut.cz}%

\maketitle

\paragraph{Abstract}

We prove the existence
of unique regular solutions of steady-state buoyancy-driven flows
of viscous incompressible heat-conducting fluids
in 3D open channels under mixed boundary conditions.
The model takes into account
the phenomena of the viscous energy dissipation.

\section{Introduction}
\label{sec:1}
Let $\om$ be a bounded domain in $\mathbb{R}^3$ with
boundary $\partial \Omega$, $\Gamma_D$ and
$\Gamma_N$ are ${C}^{\infty}$-smooth open
disjoint subsets of $\pom$ such that $\pom =
\overline{\Gamma}_D\cup\overline{\Gamma}_N$,
$\Gamma_D\neq\emptyset$, $\Gamma_N\neq\emptyset$, $\mathcal{M} =
\partial\om-(\Gamma_D\cup\Gamma_N) = \overline{\Gamma}_D\cap\overline{\Gamma}_N =
\bigcup_{j\in \mathcal{J}} \mathcal{M}_i$, $\mathcal{J} =
\left\{1,\dots,d\right\}$, and the $2$--dimensional measure of
$\mathcal{M}$ is zero and $\mathcal{M}_i$ are smooth nonintersecting curves {(this
means that $\mathcal{M}_i$ are smooth curved nonintersecting edges and vertices
(conical points) on
$\partial \Omega$ are excluded)}. Moreover, all portions
of $\Gamma_N$ are taken to be flat
and $\Gamma_D$ and $\Gamma_N$ form an angle
$\omega_{\mathcal{M}}=\pi/2$ at all points of $\mathcal{M}$ (in the
sense of tangential planes).
{ In a physical sense, $\Omega$ represents a
``truncated'' region of an unbounded channel system occupied by a
moving heat-conducting viscous incompressible fluid. $\Gamma_D$
will denote the ``lateral'' surface and $\Gamma_N$ represents the
open parts of the channel $\Omega$.
In addition, we assume that in/outflow channel
segments extend as straight pipes (see Figure \ref{channel}).
}

The strong formulation of our problem reads as follows:
\begin{align}
\varrho_0(\bfu\cdot\nabla)\bfu - \nu\Delta \bfu + \nabla P
&=
\varrho(\theta)\bfg
&& \textmd{in}\;\om,
\label{nse}
\\
\nabla \cdot (\varrho_0\bfu) &= 0
&& \textmd{in}\;\om,
\label{eq:cont}
\\
c_V \varrho(\theta) \bfu\cdot\nabla\theta  - \lambda \Delta \theta
&=
\alpha_1\nu \bfe(\bfu):\bfe(\bfu)
&& \textmd{in}\;\om,
\label{heat_eq}
\\
\bfu &= {\bf 0} &&\textmd{on}\;\Gamma_{D},
\label{bc_dirichlet_u}
\\
\theta &= \theta_D &&\textmd{on}\;\Gamma_{D},
\label{bc_dirichlet_theta}
\\
-P\bfn + \nu (\nabla \bfu) \; \bfn &= {\bf0}
&&\textmd{on}\;\Gamma_{N},
\label{bc_do_nothing}
\\
\nabla\theta \cdot \bfn &= 0
&&\textmd{on}\;\Gamma_{N}.
\label{bc_neumann}
\end{align}
Equations \eqref{nse}--\eqref{heat_eq} represent the
balance equations for linear
momentum, mass and internal energy of the homogeneous fluid
and the
system \eqref{nse}--\eqref{bc_neumann} describes
stationary buoyancy-driven flows of viscous
incompressible heat-conducting fluids with dissipative
heating in the open channel $\Omega$.
\begin{figure}[h]
  \includegraphics[angle=0,width=6cm]{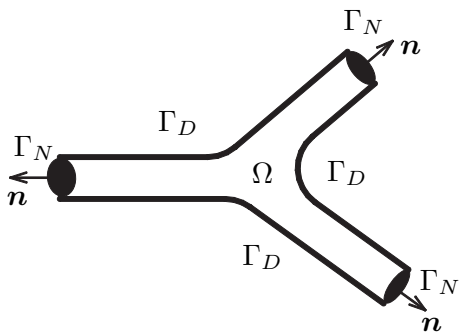}
\caption{$\Omega$ represents a
``truncated'' region of an unbounded channel system occupied by a fluid.}
\label{channel}
\end{figure}
In the model, $\bfu=(u_1,u_2,u_3)$, $P$ and $\theta$ denote the unknown
velocity, pressure and temperature, respectively. Tensor
$\bfe(\bfu)$ denotes the symmetric part of the velocity gradient
$
\bfe(\bfu)=[\nabla\bfu+(\nabla\bfu)^{\top}]/2.
$
$\bfn$ denotes the
unit outward normal with respect to $\Omega$ along $\partial
\Omega$.
Data of the problem are as follows: $\bfg$ is a body force
and  $\theta_D$ is a given function representing the
prescribed distribution of the
temperature $\theta$ on $\Gamma_D$.
Positive constant material coefficients represent
the kinematic viscosity $\nu$, reference density $\varrho_0$, heat conductivity
$\lambda$ and the specific heat at constant volume $c_V$.
Since most of the fluids, especially liquids,
are slightly compressible,
we consider the fluid to be ``mechanically incompressible'',
yet ``thermally expansible''.
{ Following the well-known Boussinesq approximation,
the temperature dependent density
is used in the energy equation \eqref{heat_eq}
and to compute the buoyancy force $\varrho(\theta)\bfg$
on the right-hand side of equation \eqref{nse}. Everywhere else in the model,
$\varrho$ is replaced by the reference value $\varrho_0$.}
Change of
density $\varrho$ with temperature is given by
strictly positive, nonincreasing and continuous function satisfying
\begin{equation}\label{con:rho}
0 <   \varrho(\xi) \leq \varrho^{\sharp} < +\infty \quad \forall \xi \in \mathbb{R}
\quad (\varrho^{\sharp}={\rm const > 0}).
\end{equation}

Coefficient $\alpha_1$ reflects the dissipation effect, which
is omitted frequently in many mathematical models
\cite{lorca1996,Rajagopal1996,Rocha2003,Roa2006}.
Otherwise, taking into account the dissipative term
$\alpha_1\nu \bfe(\bfu):\bfe(\bfu)$
with quadratic growth of the gradient,
the equations \eqref{nse}--\eqref{bc_neumann} represent
the elliptic system with strong nonlinearities without
appropriate general existence and regularity theory.
In \cite{frehse}, Frehse presented an example of a discontinuous
bounded weak solution  of a nonlinear elliptic system
of the type $\Delta \bfu = \bfF(\bfu,\nabla\bfu)$, where $\bfF$
is smooth and has quadratic growth in $\nabla\bfu$.

{
It is matter of discussion which boundary condition
should be prescribed on the outlets of the channel system. The
boundary condition (\ref{bc_do_nothing}), introduced originally
in \cite{glowinski,Gresho,Hey,SaniGresho}  and which is often called the ``do nothing''
boundary condition, results from a variational principle and has
been proven to be convenient in numerical modeling of
parallel flows \cite{Hey,liu}.
The ``do nothing'' boundary condition is often used
in the computational simulation of blood flow in the human vascular system
(see e.g.  \cite{GaRaRoTu}).
Let us mention some other interesting problems.
Because of the ``do nothing'' boundary condition
\eqref{bc_do_nothing},
some uncontrolled ``backward flow'' can take place
at the outlets of the channel
and we are not able to prove an ``a priori'' estimate for the convective
terms in the system  (cf. \cite{KraNeu5}).
Consequently, only local solutions can be proven and
the question of whether a
given solution is unique
is, to date, open (even for ``small data'').
This makes the present problem quite different than in the
case of Dirichlet-type boundary conditions for $\bfu$ on the whole boundary
frequently studied in the literature
(cf. \cite{charki,lorca1996,Paloka2009,Paloka2009b,Rocha2003,Roa2006}).
}

In \cite{BenesKucera2012a}, the authors proved
the existence of the local weak solutions
of the system \eqref{nse}--\eqref{bc_neumann}
(with constant density $\varrho$)
in a 3D open cylindrical channel with the prescribed ``free surface'' boundary condition
$
-P\bfn+\nu[\nabla\bfu+(\nabla\bfu)^{\top}]\bfn={\bf0}
$
on the output of the channel. In \cite{Benes2012b},
the author proved the $W^{2,8/7}$-regularity for the velocity and
temperature of the problem in 2D Lipschitz domains and
various types of boundary conditions.
%
In this work we provide an existence proof of
$W^{2,s}$-regular solutions of steady-state buoyancy-driven flows in 3D
open channels modeled by the problem \eqref{nse}--\eqref{bc_neumann},
which does not require small data for $\theta$ on $\Gamma_D$
described by the function $\theta_D$. Finally, the uniqueness of the solution
is discussed in this paper.

{
The paper is organized as follows.
In Section
\ref{sec:preliminaries}, we introduce basic notations and some
appropriate function spaces in order to precisely formulate our
problem and prove some auxiliary Lemmas which will be used
in the proof of the main result.
In Section \ref{sec:main_result}, we formulate the problem in a
variational setting under the framework of free divergence functional spaces
and establish the main result of our work. The main advantage
of the formulation of the Navier-Stokes system in free divergence spaces
is the elimination of the pressure $P$ to consider only the couple $\bfu$ and $\theta$
as the primary unknowns of the coupled system.
The main result is proved in Section~\ref{sec:proof_main}.
In Subsection \ref{sec:main_result_existence},
we prove the existence of the solution
introducing iterative scheme to uncouple the system.
Let us briefly describe the rough idea of the proof.
Introducing the translated function $\vartheta_0=\theta - \theta_D$
we solve the corresponding problem with homogeneous boundary conditions.
For given temperature, say $\vartheta_0$, in the buoyancy term
on the right hand side in \eqref{nse} we find $\bfu$,
the solution of the decoupled Navier-Stokes equations \eqref{nse}--\eqref{eq:cont}
with mixed boundary conditions \eqref{bc_dirichlet_u}--\eqref{bc_do_nothing}
via the Banach contraction principle. Now with $\bfu$ in hand
we modify \eqref{heat_eq}
substituting $\vartheta_0$ and $\bfu$ into convective and dissipative terms and
 find $\vartheta$, the solution of the linearized heat (Poisson) equation
with the mixed boundary conditions \eqref{bc_dirichlet_theta} and \eqref{bc_neumann}.
Finally we show that the map $\vartheta_0 \rightarrow \vartheta$
is completely continuous and maps some ball into itself.
Hence the existence of at least one solution follows from the Leray-Schauder theorem.
In Subsection \ref{sec:main_result_uniqueness}, the uniqueness of the solution
is established under the assumption of Lipschitz
continuity of $\varrho$  .
}

\section{Preliminaries}
\label{sec:preliminaries}
Throughout the paper, we will always use positive constants $c$,
$c_1$, $c_2$, $\dots$, which are not specified and which may differ
from line to line, however, do not depend on
the functions under consideration.

Let
\begin{eqnarray*}
{\bfC}_{\sigma,D}^{\infty}
&:=&
\left\{\bfu\in C^\infty(\overline{\Omega})^3; \,
\textmd{div}\,\bfu = 0, \, {\textmd{supp}\, \bfu}  \cap \Gamma_D =
\emptyset  \right\},
\\
C_D^{\infty}
&:=&
\left\{\theta \in
C^\infty(\overline{\Omega}); \, {\textmd{supp}\, \theta }  \cap
\Gamma_D = \emptyset\right\}
\end{eqnarray*}
and $\mathbf{V}_{\sigma,D}^{k,p}$ be the closure
of ${\bfC}_{\sigma,D}^{\infty}$ in the
norm of $W^{k,p}(\om)^3$, $k\ge 0$ and $1\leq p \leq \infty$.
Similarly, let $V_D^{k,p}$ be a closure of
$C_D^{\infty}$ in the norm of $W^{k,p}(\om)$. Then
$\mathbf{V}_{\sigma,D}^{k,p}$ and $V_D^{k,p}$, respectively, are
Banach spaces with the norms of the spaces $W^{k,p}(\om)^3$ and
$W^{k,p}(\om)$, respectively.

We
suppose that $r,s \in \mathbb{R}$
are fixed numbers throughout the paper such that
$s \in [4/3,s_0)$, $ s_0 = 3 + \epsilon$
($\epsilon>0$ sufficiently small) and
\begin{equation}\label{s_r}
  \left\{
\begin{array}{lcl}
r \in \left[6/5;\frac{3s}{2(3-s)}\right] & {\rm for} & s \in [4/3,3),
\\
r \in [6/5;+\infty) & {\rm for} &  s \in [3,s_0)
\end{array} \right.
\end{equation}
(the value $s_0$ will be clarified later).

To simplify mathematical formulations we introduce the following
notations:
\begin{align}
a(\bfu,\bfv) & := \nu\iom \nabla\bfu : \nabla\bfv \,{\rm d}\Omega,  
\label{form_a}
\\
b(\bfu,\bfv,\bfw)& := \iom \varrho_0(\bfu\cdot\nabla)\bfv \cdot \bfw \,{\rm d}\Omega,
\label{form_b}
\\
\kappa(\theta,\varphi) & :=  \lambda\iom \nabla \theta \cdot \nabla
\varphi  \,{\rm d}\Omega,
\label{form_c}
\\
d(\vartheta,\bfu,\theta,\varphi)& := c_V  \iom \varrho(\vartheta)\bfu\cdot\nabla\theta
 \,  \varphi  \,{\rm d}\Omega,
 \label{form_d}
\\
e(\bfu,\bfv,\varphi) & :=  \alpha_1\nu \iom  \bfe(\bfu):\bfe(\bfv) \varphi
 \,{\rm d}\Omega,
\label{form_e}
 \\
(\bfu,\bfv) &:=  \iom \bfu \cdot \bfv  \,{\rm d}\Omega,
 \label{scalar_Lu}
\\
(\theta,\varphi)_{\Omega} & :=  \iom \theta \varphi  \,{\rm d}\Omega.
\label{scalar_Lt}
\end{align}
In \eqref{form_a}--\eqref{scalar_Lt} all functions
$\bfu,\bfv,\bfw,\theta,\vartheta,\varphi$ are smooth enough, such that all
integrals on the right-hand sides make sense.

Further, define the spaces
\begin{equation}\label{D_a}
\mathbf{D}_{a}^s := \left\{\bfu \; |\; \bff \in \mathbf{V}_{\sigma,D}^{0,s},\,
a(\bfu,\bfv) =(\bff,\bfv)  \textmd{ for all } \bfv\in
\mathbf{V}_{\sigma,D}^{1,2} \right\}
\end{equation}
and
\begin{equation}\label{D_kappa}
{D}^r_{\kappa} :=  \left\{\theta \; |\; f \in
V_D^{0,r}, \, \kappa(\theta,\varphi)
=(f,\varphi)_{\Omega} \textmd{ for all } \varphi \in
V_D^{1,2} \right\},
\end{equation}
equipped with the norms
\begin{equation}\label{def_norms}
\|\bfu\|_{\mathbf{D}_{a}^s} := \|\bff\|_{\mathbf{V}_{\sigma,D}^{0,s}}
\quad\textmd{ and }\quad
\|\theta\|_{D^r_{\theta}} := \|f\|_{V_D^{0,r}},
\end{equation}
where $\bfu$ and $\bff$ are corresponding functions via \eqref{D_a}
and $\theta$ and $f$ are corresponding functions via
\eqref{D_kappa}.




\begin{lem}\label{emb_D}
For given $s,r \in \mathbb{R}$ satisfying \eqref{s_r}
the following embeddings hold:
\begin{align}
\mathbf{D}^s_a &\hookrightarrow \mathbf{W}^{2,s},
&
\|\bfu\|_{\mathbf{W}^{2,s}} \leqslant c(\nu,\Omega)\|\bfu\|_{\mathbf{D}_{a}^s}
&
\quad \forall \bfu \in \mathbf{D}^s_a,
\label{emd_u}
\\
D^r_{\kappa}  &\hookrightarrow {W}^{2,r}(\Omega),
&
\|\theta\|_{{W}^{2,r}(\Omega)} \leqslant c(\lambda,\Omega)\|\theta\|_{D^r_{\kappa}}
&
\quad \forall \theta \in D^r_{\kappa}.
\label{emd_theta}
\end{align}
\end{lem}
{
\begin{proof}
It is well known that for every $\bff \in (\mathbf{V}_{\sigma,D}^{1,2})^*$ there
exists the uniquely determined $\bfu \in \mathbf{V}_{\sigma,D}^{1,2}$ such
that $a(\bfu,\bfv)   = \langle\bff ,
\bfv \rangle$ for every $\bfv \in \mathbf{V}_{\sigma,D}^{1,2}$
and
$$
\|\bfu\|_{\mathbf{V}_{\sigma,D}^{1,2}}
\leqslant
c(\nu,\Omega)\|\bff\|_{(\mathbf{V}_{\sigma,D}^{1,2})^*}.
$$

Similarly, for every $f \in (V_D^{1,2})^*$ there
exists the uniquely determined $\theta \in V_D^{1,2}$ such
that $\kappa(\theta,\varphi) = \langle f,
\varphi \rangle$ for every $\varphi \in V_D^{1,2}$
and
$$
\|\theta\|_{V_D^{1,2}} \leqslant c(\lambda,\Omega)\|f\|_{(V_D^{1,2})^*}.
$$
Using known embedding for Sobolev spaces
$W^{1,2}(\Omega) \hookrightarrow L^6(\Omega)$
 we have
$L^{6/5}(\Omega)\hookrightarrow W^{1,2}(\Omega)^*$
and therefore we deduce
$\mathbf{V}_{\sigma,D}^{0,s}\hookrightarrow (\mathbf{V}_{\sigma,D}^{1,2})^*$
and
$V_D^{0,r}\hookrightarrow (V_D^{1,2})^*$ (with $r$ and $s$ satisfying \eqref{s_r})
and finally
$\mathbf{D}_{a}^s \hookrightarrow \mathbf{V}_{\sigma,D}^{1,2}$
and
$D^r_{\theta}\hookrightarrow V_D^{1,2}$.
Now higher smoothness of $\bfu$ and $\theta$ together with the estimates
\eqref{emd_u}--\eqref{emd_theta}
follow from the regularity results
for the Poisson equation and the Stokes system in $\Omega$ with the
mixed boundary conditions.
The proof is rather technical and therefore
postponed to Appendix \ref{app:regularity}.
\end{proof}
}


\begin{rem}
Throughout the paper $\varepsilon$ denotes a sufficiently small positive real number.
\end{rem}

\begin{lem}\label{lem:sup_estimates}
There exist numbers $C_b$, $C_d$ and $C_e$ such that for every
 $\bfu$, $\bfv \in D^s_a$ and $\theta\in W^{2-\varepsilon,r}(\Omega)$ we have
\begin{eqnarray}
\|b(\bfu,\bfv,\cdot)\|_{\mathbf{V}_{\sigma,D}^{0,s}}
&\leqslant&
\varrho_0
C_b(\nu,\Omega)
\|\bfu\|_{\mathbf{D}^s_a}  \|\bfv\|_{\mathbf{D}^s_a},
\label{est_b}
\\
\| d(\theta,\bfu,\theta,\cdot) \|_{V_{D}^{0,r}}
&\leqslant&
c_V \varrho^{\sharp}
C_d(\nu,\Omega)
\|\bfu\|_{\mathbf{D}^s_a}
\|\theta\|_{W^{2-\varepsilon,r}(\Omega)},
\label{est_d}
\\
\|  e(\bfu,\bfv,\cdot) \|_{V_{D}^{0,r}}
&\leqslant&
\alpha_1
\nu
C_e(\nu,\Omega)
\|\bfu\|_{\mathbf{D}^s_a}  \|\bfv\|_{\mathbf{D}^s_a}.
\label{est_e}
\end{eqnarray}
\end{lem}
\begin{proof}
By H\"{o}lder inequality, Sobolev embeddings
(see e.g. \cite{AdamsFournier1992,KufFucJoh1977,lionsmagenes})
and using \eqref{s_r}, \eqref{emd_u} and \eqref{emd_theta}
we arrive at the estimates
\begin{eqnarray*}
\|b(\bfu,\bfv,\cdot)\|_{\mathbf{V}_{\sigma,D}^{0,s}}
&\leqslant&
\varrho_0
\|\bfu\|_{\mathbf{V}_{\sigma,D}^{0,2s}}
\|\nabla\bfv\|_{\mathbf{V}_{\sigma,D}^{0,2s}}
\\
&\leqslant&
\varrho_0
c(\Omega)
\|\bfu\|_{\mathbf{V}_{\sigma,D}^{2,s}}
\|\bfv\|_{\mathbf{V}_{\sigma,D}^{2,s}}
\\
&\leqslant&
\varrho_0
C_b(\nu,\Omega)
\|\bfu\|_{\mathbf{D}^s_a}  \|\bfv\|_{\mathbf{D}^s_a}.
\end{eqnarray*}
Further,
\begin{eqnarray*}
\| d(\theta,\bfu,\theta,\cdot) \|_{V_{D}^{0,r}}
&\leqslant&
c_V \varrho^{\sharp}
\|\bfu\|_{\mathbf{V}_{\sigma,D}^{0,2r}}
\|\nabla\theta\|_{L^{2r}(\Omega)^3}
\\
&\leqslant&
c_V \varrho^{\sharp}
c(\Omega)
\|\bfu\|_{\mathbf{V}_{\sigma,D}^{2,s}}
\|\theta\|_{W^{2-\varepsilon,r}(\Omega)}
\\
&\leqslant&
c_V \varrho^{\sharp}
C_d(\nu,\Omega)
\|\bfu\|_{\mathbf{D}^s_a}
\|\theta\|_{W^{2-\varepsilon,r}(\Omega)}
\end{eqnarray*}
and finally
\begin{eqnarray*}
\|  e(\bfu,\bfv,\cdot) \|_{V_{D}^{0,r}}
&\leqslant&
\alpha_1
\nu
\|\nabla\bfu\|_{\mathbf{V}_{\sigma,D}^{0,2r}}
\|\nabla\bfv\|_{\mathbf{V}_{\sigma,D}^{0,2r}}
\\
&\leqslant&
\alpha_1
\nu
c(\Omega)
\|\bfu\|_{\mathbf{V}_{\sigma,D}^{2,s}}
\|\bfv\|_{\mathbf{V}_{\sigma,D}^{2,s}}
\nonumber
\\
&\leqslant&
\alpha_1
\nu
C_e(\nu,\Omega)
\|\bfu\|_{\mathbf{D}^s_a}  \|\bfv\|_{\mathbf{D}^s_a}.
\end{eqnarray*}
\end{proof}

\begin{lem}\label{continuity_d}
Let $\bfu \in D^s_a$,
$\theta \in W^{2-\varepsilon,r}(\Omega)$
and
$\theta_n$
be a sequence in $W^{2-\varepsilon,r}(\Omega)$
such that $\theta_n \rightarrow \theta$ in $W^{2-\varepsilon,r}(\Omega)$.
Then
\begin{equation}\label{conv_d}
\|d(\theta_n,\bfu,\theta,\cdot)-d(\theta,\bfu,\theta,\cdot)\|_{V_{D}^{0,r}}
\rightarrow 0.
\end{equation}
\end{lem}
\begin{proof}
{
Let $\bfu \in D^s_a$ and $\theta \in W^{2-\varepsilon,r}(\Omega)$.
Then
\begin{eqnarray*}
\| \bfu\cdot\nabla\theta \|_{L^r(\Omega)}
&\leqslant&
\|\bfu\|_{\mathbf{V}_{\sigma,D}^{0,2r}}
\|\nabla\theta\|_{L^{2r}(\Omega)^3}
\\
&\leqslant&
c(\Omega)
\|\bfu\|_{\mathbf{V}_{\sigma,D}^{2,s}}
\|\theta\|_{W^{2-\varepsilon,r}(\Omega)}
\\
&\leqslant&
C_d(\nu,\Omega)
\|\bfu\|_{\mathbf{D}^s_a}
\|\theta\|_{W^{2-\varepsilon,r}(\Omega)},
\end{eqnarray*}
which yields $\bfu\cdot\nabla\theta \in L^r(\Omega)$.
Hence,
the function $f$ defined by the formula (recall \eqref{con:rho})
$$
f(\bfx,\xi)= c_V\varrho(\xi)\bfu(\bfx)\cdot\nabla\theta(\bfx)
$$
has the Carath\'{e}odory property (see \cite[Definition 12.2]{FucKuf}), i.e.
(i) for all $\xi \in \mathbb{R}$, the function $f_{\xi}(\bfx)=f(\bfx,\xi)$
(as function of the variable $\bfx$)
is measurable on $\Omega$
and (ii) for almost all $\bfx \in \Omega$, the function
$f_{\bfx}(\xi)=f(\bfx,\xi)$
(as function of the variable $\xi$)
continuous on $\mathbb{R}$
(by continuity of $\varrho$).

Define the so-called N\v{e}mytski\u{\i} operator,
$\mathcal{N}: L^r(\Omega) \rightarrow L^r(\Omega)$,
by the formula
$$
\mathcal{N}(\theta)(\bfx)=f(\bfx,\theta).
$$
By \cite[Theorem 12.10]{FucKuf} we deduce that the N\v{e}mytski\u{\i} operator $\mathcal{N}$
is a continuous operator from  $L^r(\Omega)$ into $L^r(\Omega)$  (recall \eqref{con:rho}).
Since we assume $\theta_n \rightarrow \theta$ in $W^{2-\varepsilon,r}(\Omega)$
and $W^{2-\varepsilon,r}(\Omega) \hookrightarrow L^r(\Omega)$,
we get \eqref{conv_d}.
}
\end{proof}
\begin{rem}\label{comp_emb}
By Lemma \ref{emb_D} and the standard compact embedding of Sobolev spaces
the embedding $D^{r}_{\kappa} \rightarrow V^{2-\varepsilon,r}_{D}$ is compact and
we denote by $C_{\varepsilon}$ the embedding constant such that
\begin{displaymath}
\|\theta\|_{V^{2-\varepsilon,r}_{D}} \leqslant C_{\varepsilon}\|\theta\|_{D^r_{\kappa}}
\quad \forall \theta \in D^r_{\kappa}.
\end{displaymath}
\end{rem}

\section{The main result}
\label{sec:main_result}
Our problem reads as follows:
for given  $\bfg\in \mathbf{V}_{\sigma,D}^{0,s}$ and $\theta_D \in W^{2,r}(\Omega)$
find a couple $[\bfu,\theta]$ such that $\bfu \in D^s_a$,
$\theta \in \theta_D + D^r_{\kappa}$ and the following system
\begin{eqnarray}
a(\bfu,\bfv) + b(\bfu,\bfu,\bfv)
&=&
(\varrho(\theta)\bfg,\bfv),
\label{weak_NS}
 \\
\kappa(\theta,\varphi)+
d(\theta,\bfu,\theta,\varphi)
&=&
e(\bfu,\bfu,\varphi)
\label{weak_heat}
\end{eqnarray}
holds for every $[\bfv,\varphi] \in  \mathbf{V}_{\sigma,D}^{1,2} \times V_D^{1,2}$.
The couple $[\bfu,\theta]$ will be called the strong solution to the system \eqref{nse}--\eqref{bc_neumann}.
\begin{thm}[Main result]\label{theorem_main_result}
(i)
Let  $\bfg\in \mathbf{V}_{\sigma,D}^{0,s}$ and $\theta_D \in W^{2,r}(\Omega)$
and assume that
\begin{equation}\label{assumption}
\|\bfg\|_{\mathbf{V}_{\sigma,D}^{0,s}} \leqslant \frac{\beta}{4C_b\varrho^{\sharp}\varrho_0}
<
\frac{1}{2 C_{\varepsilon} C_d c_V (\varrho^{\sharp})^2}
\end{equation}
with some $\beta \in (0, 1)$. Then there exists the strong solution to the system
\eqref{nse}--\eqref{bc_neumann}.

(ii) Let $r>3/2$ and, in addition, $\varrho$
be Lipschitz continuous, i.e.
$$
|\varrho(\zeta_1) - \varrho(\zeta_2)| \leq C_{\varrho} |\zeta_1-\zeta_2| \quad \forall \zeta_1,\zeta_2 \in \mathbb{R} \quad
(C_{\varrho}={\rm const > 0}).
$$
Then there exists $\gamma>0$ such that, if $[\bfu,\theta]$ is the strong solution of the
system \eqref{nse}--\eqref{bc_neumann}, $\bfu \in D^s_a$,
$\theta \in \theta_D + D^r_{\kappa}$ and satisfying
$(\|\theta\|_{W^{2,r}(\Omega)}+\|\bfu\|_{\mathbf{D}^s_a}) < \gamma$,
then it is unique.
\end{thm}
\begin{rem}
Let us note that the assumption $r>3/2$
in Theorem \ref{theorem_main_result}(ii) ensures $\theta \in L^{\infty}(\Omega)$,
which is required in the proof of uniqueness.
\end{rem}

\section{Proof of the main result}\label{sec:proof_main}

\subsection{Existence of strong solutions}
\label{sec:main_result_existence}

For an arbitrary fixed
$[\bfu_0,\vartheta_0]\in \mathbf{D}^s_a \times V^{2-\varepsilon,r}_{D}$  we now
consider the following auxiliary problem: to find a couple
$[\bfw,\vartheta]\in\mathbf{D}^s_a \times D^r_{\kappa}$,
such that
\begin{eqnarray}
a(\bfw,\bfv)
&=&
(\varrho(\vartheta_0+\theta_D)\bfg,\bfv) - b(\bfu_0,\bfu_0,\bfv),
\label{linear_problem_1}
\\
\kappa(\vartheta,\varphi)
&=&
e(\bfw,\bfw,\varphi)-d(\vartheta_0+\theta_D,\bfw,\vartheta_0+\theta_D,\varphi)
-\kappa(\theta_D,\varphi)
\label{linear_problem_2}
\end{eqnarray}
for every $[\bfv,\varphi] \in  \mathbf{V}_{\sigma,D}^{1,2} \times V_D^{1,2}$.

First we prove some estimates for terms on the right hand
sides in \eqref{linear_problem_1}--\eqref{linear_problem_2}.
For an arbitrary $[\bfu_0,\vartheta_0]\in \mathbf{D}^s_a \times V_D^{2-\varepsilon,r}$ we arrive at
\begin{eqnarray}
\|(\varrho(\vartheta_0+\theta_D)\bfg,\cdot)\|_{\mathbf{V}_{\sigma,D}^{0,s}}
& \leqslant &
\varrho^{\sharp}\|\bfg\|_{\mathbf{V}_{\sigma,D}^{0,s}},
\label{est01a}
\\
\|b(\bfu_0,\bfu_0,\cdot)\|_{\mathbf{V}_{\sigma,D}^{0,s}}
& \leqslant &
\varrho_0 C_b \|\bfu_0\|^2_{\mathbf{D}^s_a},
\label{est02}
\end{eqnarray}
where $C_b=C_b(\nu,\Omega)$ (see \eqref{est_b}).
The inequalities \eqref{est01a}--\eqref{est02} together with \eqref{def_norms}
yield the estimate for the unique solution $\bfw \in \mathbf{D}^s_a$ of the problem \eqref{linear_problem_1}
\begin{equation}\label{est03}
\|\bfw\|_{\mathbf{D}^s_a}
\leqslant
\varrho^{\sharp}\|\bfg\|_{\mathbf{V}_{\sigma,D}^{0,s}}
+ \varrho_0 C_b \|\bfu_0\|^2_{\mathbf{D}^s_a}.
\end{equation}
Now having $\bfw \in \mathbf{D}^s_a$ and $\vartheta_0 \in V_D^{2-\varepsilon,r}$,
for all terms on the right hand side of \eqref{linear_problem_2} we arrive at the estimates
\begin{eqnarray}
\| \kappa(\theta_D,\cdot) \|_{V_D^{0,r}}
&\leqslant&
c \|\theta_D\|_{W^{2,r}(\Omega)},
\label{est05a}
\\
\| e(\bfw,\bfw,\cdot)  \|_{V_D^{0,r}}
& \leqslant &
c \|\bfw\|^2_{\mathbf{D}^s_a}
\label{est05c}
\end{eqnarray}
and
\begin{equation}\label{est05d}
\|d(\vartheta_0+\theta_D,\bfw,\vartheta_0+\theta_D,\cdot) \|_{V_D^{0,r}}
\leqslant
C_d \varrho^{\sharp} c_V  \|\bfw\|_{\mathbf{D}^s_a}
\left( \|\vartheta_0\|_{V_D^{2-\varepsilon,r}}+\|\theta_D\|_{W^{2,r}} \right).
\end{equation}
Now the inequalities \eqref{est05a}--\eqref{est05d} together with \eqref{def_norms}
yield the estimate for the unique solution $\vartheta \in {D^r_{\kappa}}$ of the problem \eqref{linear_problem_2}
\begin{eqnarray}\label{est09}
\|\vartheta\|_{V^{2-\varepsilon,r}_{D}}
\leqslant
C_{\varepsilon} \|\vartheta\|_{D^r_{\kappa}}
&
\leqslant
&
 C_{\varepsilon}  C_d \varrho^{\sharp} c_V  \|\bfw\|_{\mathbf{D}^s_a}
\|\vartheta_0\|_{V^{2-\varepsilon,r}_{D}}
\nonumber
\\
&&
+ c_1 C_{\varepsilon} \|\theta_D\|_{W^{2,r}(\Omega)}
+ c_2 C_{\varepsilon} \|\bfw\|^2_{\mathbf{D}^s_a}
\nonumber
\\
&&
+
C_{\varepsilon} C_d \varrho^{\sharp} c_V  \|\bfw\|_{\mathbf{D}^s_a}
\|\theta_D\|_{W^{2,r}(\Omega)}.
\end{eqnarray}

For a given couple $[\bfu_0,\vartheta_0]\in \mathbf{D}^s_a \times V_D^{2-\varepsilon,r}$
let $\bfw \in \mathbf{D}^s_a$ be the unique solution of the equation \eqref{linear_problem_1}.
Fix $\vartheta_0 \in  V_D^{2-\varepsilon,r}$ and consider the map
\begin{center}
$
\mathcal{K}_{\vartheta_0}:  \mathbf{D}^s_a   \rightarrow  \mathbf{D}^s_a  $
\qquad with $\mathcal{K}_{\vartheta_0}(\bfu_0)=\bfw.
$
\end{center}
\begin{lem}\label{contraction}
Operator $\mathcal{K}_{\vartheta_0}$ realizes contraction in the closed ball
($\beta$ is the constant from \eqref{assumption})
$$
M=\left\{\bfv\in \mathbf{D}^s_a; \;
\|\bfv\|_{\mathbf{D}^s_a} \leqslant \frac{\beta}{2 C_b \varrho_0} \right\}.
$$
\end{lem}
\begin{proof}
Using the estimate \eqref{est03} and the assumption \eqref{assumption} we get
$\mathcal{K}_{\vartheta_0}(M)\subset M$. Let $\bfu_0$ and $\bar{\bfu}_0 \in M$.
Then by \eqref{est_b} we arrive at
\begin{eqnarray*}
\|\mathcal{K}_{\vartheta_0}(\bfu_0) - \mathcal{K}_{\vartheta_0}(\bar{\bfu}_0) \|_{\mathbf{D}^s_a}
&=&
\| b(\bfu_0,\bfu_0,\cdot) - b(\bar{\bfu}_0,
\bar{\bfu}_0,\cdot) \|_{\mathbf{V}_{\sigma,D}^{0,s}}
\nonumber
\\
&\leqslant& \varrho_0 C_b
(\|\bfu_0\|_{\mathbf{D}^s_a} + \|\bar{\bfu}_0\|_{\mathbf{D}^s_a}) \|\bfu_0 - \bar{\bfu}_0\|_{\mathbf{D}^s_a}
\nonumber
\\
&<& \beta\|\bfu_0 - \bar{\bfu}_0\|_{\mathbf{D}^s_a}
\end{eqnarray*}
with $\beta\in(0,1)$ (cf. \eqref{assumption}).
The proof is complete.
\end{proof}
As a consequence of Lemma \ref{contraction} and the Banach fixed point theorem there exists
the unique $\bfw \in M$
such that $\mathcal{K}_{\vartheta_0}(\bfw)=\bfw$. Define the operator
$\mathcal{T}_1: V^{2-\varepsilon,r}_{D} \rightarrow M$
by $\mathcal{T}_1(\vartheta_0)=\bfw$.
Let $\vartheta = \mathcal{T}_2(\mathcal{T}_1(\vartheta_0),\vartheta_0)$,
$\vartheta\in D^{r}_{\kappa} \hookrightarrow V^{2-\varepsilon,r}_{D}$,
be the solution of the problem \eqref{linear_problem_2}.
If there exists a fixed point of $\mathcal{T}_2$
(again denoted by $\vartheta$), then $\theta=\mathcal{T}_2(\vartheta)+\theta_D$
and $\bfu = \mathcal{T}_1(\vartheta)$
solve the system \eqref{weak_NS}--\eqref{weak_heat}.

Let us note that the ball $M$ is independent of the choice of $\vartheta_0 \in V^{2-\varepsilon,r}_{D}$
and the right hand side of the inequality \eqref{est09} depends linearly on $\|\vartheta_0\|_{V^{2-\varepsilon,r}_{D}}$. Moreover,
for $\bfw \in M$ and taking into account the assumption \eqref{assumption}
we obtain
\begin{displaymath}
C_{\varepsilon} C_d \varrho^{\sharp} c_V  \|\bfw\|_{\mathbf{D}^s_a}
\leqslant
C_{\varepsilon} C_d  \varrho^{\sharp} c_V\frac{\beta}{2 C_b \varrho_0}
<1.
\end{displaymath}
Consequently, there exists sufficiently large $R$ such that $\mathcal{T}_2$ maps the ball
\begin{displaymath}
B=\left\{\phi \in V^{2-\varepsilon,r}_{D};
\;  \|\phi\|_{V^{2-\varepsilon,r}_{D}} \leqslant R \right\}
\end{displaymath}
into itself. By the compact embedding $D^{r}_{\kappa} \hookrightarrow \hookrightarrow V^{2-\varepsilon,r}_{D}$
the operator $\mathcal{T}_2$ is completely continuous if we prove that $\mathcal{T}_2$ is continuous.

Let $\vartheta_0 \in V^{2-\varepsilon,r}_{D}$
and $(\vartheta_0)_n$ be a sequence in  $V^{2-\varepsilon,r}_{D}$
such that $(\vartheta_0)_n \rightarrow \vartheta_0$.
Let $\bfw=\mathcal{T}_1(\vartheta_0)$ and $\bfw_n=\mathcal{T}_1((\vartheta_0)_n)$.
By noting \eqref{def_norms} we arrive at the
estimate
\begin{eqnarray*}
\|\bfw-\bfw_n \|_{\mathbf{D}^s_a}
&\leqslant&
\| b(\bfw,\bfw,\cdot)
- b(\bfw_n,\bfw_n,\cdot) \|_{\mathbf{V}_{\sigma,D}^{0,s}}
\\
&&
+
\| (\varrho(\vartheta_0+\theta_D)\bfg,\cdot)
- (\varrho((\vartheta_0)_n+\theta_D)\bfg,\cdot)  \|_{\mathbf{V}_{\sigma,D}^{0,s}}
\\
&\leqslant& \varrho_0
C_b
(\|\bfw\|_{\mathbf{D}^s_a}
+ \|\bfw_n\|_{\mathbf{D}^s_a})
\|\bfw - \bfw_n\|_{\mathbf{D}^s_a}
\\
&&
+
\| (\varrho(\vartheta_0+\theta_D)\bfg,\cdot)
- (\varrho((\vartheta_0)_n+\theta_D)\bfg,\cdot)  \|_{\mathbf{V}_{\sigma,D}^{0,s}}
\\
&&
\!\!\!\!\!\!\!\!\!\!\!\!\!\!\!\!\!
\leqslant
\beta
\|\bfw - \bfw_n\|_{\mathbf{D}^s_a}
+
\| (\varrho(\vartheta_0+\theta_D)\bfg,\cdot)
- (\varrho((\vartheta_0)_n+\theta_D)\bfg,\cdot)  \|_{\mathbf{V}_{\sigma,D}^{0,s}}
\end{eqnarray*}
and hence
\begin{equation}\label{cont_w}
(1-\beta)\|\bfw-\bfw_n \|_{\mathbf{D}^s_a}
\leqslant
\| (\varrho(\vartheta_0+\theta_D)\bfg,\cdot)
- (\varrho((\vartheta_0)_n+\theta_D)\bfg,\cdot)  \|_{\mathbf{V}_{\sigma,D}^{0,s}},
\end{equation}
where $\beta \in (0,1)$ (given by \eqref{assumption}).
Consequently, the operator $\mathcal{T}_1: V^{2-\varepsilon,r}_{D}
\rightarrow M$
is continuous.
To prove the continuity of $\mathcal{T}_2$ let us estimate
\begin{multline*}
\|\mathcal{T}_2(\vartheta_0)-\mathcal{T}_2((\vartheta_0)_n)\|_{V^{2-\varepsilon,r}_{D}}
\leqslant
C_{\varepsilon}
\|\mathcal{T}_2(\vartheta_0)-\mathcal{T}_2((\vartheta_0)_n)\|_{D^r_{\kappa}}
\\
\leqslant C_{\varepsilon}\left(
\|e(\bfw,\bfw-\bfw_n,\cdot)\|_{V_D^{0,r}}
+ \|e(\bfw-\bfw_n,\bfw_n,\cdot)\|_{V_D^{0,r}}
\right.
\\
+\| d(\vartheta_0+\theta_D,\bfw,\vartheta_0+\theta_D,\cdot)
- d((\vartheta_0)_n+\theta_D,\bfw,\vartheta_0+\theta_D,\cdot)  \|_{V_D^{0,r}}
\\
+\| d((\vartheta_0)_n+\theta_D,\bfw_n,(\vartheta_0)_n-\vartheta_0,\cdot)  \|_{V_D^{0,r}}
+\| d((\vartheta_0)_n+\theta_D,\bfw_n-\bfw,\vartheta_0+\theta_D,\cdot)  \|_{V_D^{0,r}}
\\
\left.
+
\| (\varrho((\vartheta_0)_n+\theta_D)\alpha_2
\bfg\cdot(\bfw-\bfw_n),\cdot)_{\Omega}  \|_{V_D^{0,r}}
\right) .
\end{multline*}
By Lemma \ref{continuity_d}, estimates \eqref{est_d}, \eqref{est_e}
and \eqref{cont_w} and continuity of $\varrho$ we conclude that
$$
\|\mathcal{T}_2(\vartheta_0)-\mathcal{T}_2((\vartheta_0)_n)\|_{V^{2-\varepsilon,r}_{D}}
\rightarrow 0
$$
whenever $\|(\vartheta_0)-((\vartheta_0)_n)\|_{V^{2-\varepsilon,r}_{D}}
\rightarrow 0$. Consequently, $\mathcal{T}_2$ is completely continuous and
$\mathcal{T}_2(B)\subset B$. The existence of at least one fixed point
$\vartheta=\mathcal{T}_2(\vartheta)$ follows from the Leray-Schauder theorem.
Now the couple $[\bfu,\theta]$, $\bfu = \mathcal{T}_1(\vartheta)$
and $\theta = \theta_D + \vartheta$,
is the solution of the problem \eqref{weak_NS}--\eqref{weak_heat}.

\subsection{Uniqueness}
\label{sec:main_result_uniqueness}

Here we prove the uniqueness of the strong solution stated in
the main result. Suppose that all assumptions of Theorem
\ref{theorem_main_result} are satisfied and there are two strong
solutions $[\bfu_1,\theta_1],[\bfu_2,\theta_2]$
of the system \eqref{nse}--\eqref{bc_neumann} such that
 $\bfu_1,\bfu_2 \in D^s_a$,
$\theta_1,\theta_2 \in \theta_D + D^r_{\kappa}$.
Denote
$\bfz=\bfu_1-\bfu_2$ and $\sigma=\theta_1 - \theta_2$.
Then $\bfz$ and $\sigma$ satisfy the equations
\begin{eqnarray*}
a(\bfz,\bfv)
+
b(\bfz,\bfu_2,\bfv)
+
b(\bfu_1,\bfz,\bfv)
-
((\varrho(\theta_1)-\varrho(\theta_2))\bfg,\bfv)
&=&
0,
\\
\kappa(\sigma,\varphi)
+
d(\theta_1,\bfu_1,\theta_2,\varphi) - d(\theta_2,\bfu_1,\theta_2,\varphi)
+
d(\theta_2,\bfz,\theta_1,\varphi)
\nonumber
\\
+
d(\theta_2,\bfu_2,\sigma,\varphi)
-
e(\bfz,\bfu_1,\varphi)
-
e(\bfu_2,\bfz,\varphi)
&=&
0
\end{eqnarray*}
for every $[\bfv,\varphi] \in  \mathbf{V}_{\sigma,D}^{1,2} \times V_D^{1,2}$.
By \eqref{def_norms} we arrive at the
estimates
\begin{equation}\label{est_uni_03}
\|\bfz\|_{\mathbf{D}^s_a}
\leqslant
\|b(\bfz,\bfu_2,\cdot)\|_{\mathbf{D}^s_a}
+
\|b(\bfu_1,\bfz,\cdot)\|_{\mathbf{D}^s_a}
+
\|(\varrho(\theta_1)-\varrho(\theta_2))\bfg\|_{\mathbf{V}_{\sigma,D}^{0,s}}
\end{equation}
and
\begin{eqnarray}\label{est_uni_04}
\|\sigma\|_{D^r_{\kappa}}
&
\leqslant
&
\| d(\theta_1,\bfu_1,\theta_2,\cdot) - d(\theta_2,\bfu_1,\theta_2,\cdot) \|_{V_D^{0,r}}
+
\| d(\theta_2,\bfz,\theta_1,\cdot) \|_{V_D^{0,r}}
\nonumber
\\
&&
\| d(\theta_2,\bfu_2,\sigma,\cdot) \|_{V_D^{0,r}}
+
\| e(\bfz,\bfu_1,\cdot) \|_{V_D^{0,r}}
+
\| e(\bfu_2,\bfz,\cdot) \|_{V_D^{0,r}}.
\nonumber
\\
\end{eqnarray}

To estimate term by term on the right-hand sides of \eqref{est_uni_03}
and \eqref{est_uni_04} we use Lemma \ref{lem:sup_estimates} to obtain
\begin{eqnarray*}
\|b(\bfz,\bfu_2,\cdot)\|_{\mathbf{V}_{\sigma,D}^{0,s}}
&\leqslant&
\varrho_0
C_b
\|\bfz\|_{\mathbf{D}^s_a}  \|\bfu_2\|_{\mathbf{D}^s_a},
\\
\|b(\bfu_1,\bfz,\cdot)\|_{\mathbf{V}_{\sigma,D}^{0,s}}
&\leqslant&
\varrho_0
C_b
\|\bfu_1\|_{\mathbf{D}^s_a}  \|\bfz\|_{\mathbf{D}^s_a},
\\
\|(\varrho(\theta_1)-\varrho(\theta_2))\bfg\|_{\mathbf{V}_{\sigma,D}^{0,s}}
&\leqslant&
C_{\varrho}
\|\sigma\|_{V_D^{0,\infty}}
\|\bfg\|_{\mathbf{V}_{\sigma,D}^{0,s}}.
\end{eqnarray*}
Since $r>3/2$  we have,
using known embedding for Sobolev spaces,
$\|\sigma\|_{V_D^{0,\infty}} \leqslant C_1\|\sigma\|_{D^r_{\kappa}}$.
Further
\begin{equation*}
\| d(\theta_1,\bfu_1,\theta_2,\cdot) - d(\theta_2,\bfu_1,\theta_2,\cdot) \|_{V_D^{0,r}}
\leqslant
c_V C_1 C_{\varepsilon}
C_d
C_{\varrho}\|\sigma\|_{D^r_{\kappa}}
\|\bfu_1\|_{\mathbf{D}^s_a}
\|\theta_2\|_{W^{2,r}(\Omega)}
\end{equation*}
and finally
\begin{eqnarray*}
\| d(\theta_2,\bfz,\theta_1,\cdot) \|_{V_D^{0,r}}
&\leqslant&
c_V C_{\varepsilon} \varrho^{\sharp}
C_d
\|\bfz\|_{\mathbf{D}^s_a}
\|\theta_1\|_{W^{2,r}(\Omega)},
\\
\| d(\theta_2,\bfu_2,\sigma,\cdot) \|_{V_D^{0,r}}
&\leqslant&
c_V C_{\varepsilon} \varrho^{\sharp}
C_d
\|\bfu_2\|_{\mathbf{D}^s_a}
\|\sigma\|_{D^r_{\kappa}},
\\
\|  e(\bfz,\bfu_1,\cdot) \|_{V_{D}^{0,r}}
&\leqslant&
\alpha_1
\nu
C_e
\|\bfz\|_{\mathbf{D}^s_a}  \|\bfu_1\|_{\mathbf{D}^s_a},
\\
\|  e(\bfu_2,\bfz,\cdot) \|_{V_{D}^{0,r}}
&\leqslant&
\alpha_1
\nu
C_e
\|\bfu_2\|_{\mathbf{D}^s_a}  \|\bfz\|_{\mathbf{D}^s_a}.
\end{eqnarray*}
Hence
\begin{multline}\label{ineq_uniq_010}
\|\sigma\|_{D^r_{\kappa}}
\leqslant
\left(
c_V C_1 C_{\varepsilon} C_d C_{\varrho}
\|\bfu_1\|_{\mathbf{D}^s_a}
\|\theta_2\|_{W^{2,r}(\Omega)} +
c_V C_{\varepsilon} \varrho^{\sharp}
C_d \|\bfu_2\|_{\mathbf{D}^s_a}
\right)
\|\sigma\|_{D^r_{\kappa}}
\\
+\left(
c_V C_{\varepsilon} \varrho^{\sharp} C_d \|\theta_1\|_{W^{2,r}(\Omega)}
+
\alpha_1 \nu C_e \|\bfu_1\|_{\mathbf{D}^s_a}
+
\alpha_1
\nu
C_e
\|\bfu_2\|_{\mathbf{D}^s_a}
\right)
\|\bfz\|_{\mathbf{D}^s_a}
\end{multline}
and
\begin{multline}\label{ineq_uniq_011}
\|\bfz\|_{\mathbf{D}^s_a}
\leqslant
C_1 C_{\varrho}
\|\bfg\|_{\mathbf{V}_{\sigma,D}^{0,s}}
\|\sigma\|_{D^r_{\kappa}}
+\left(\varrho_0 C_b (\|\bfu_1\|_{\mathbf{D}^s_a}+\|\bfu_2\|_{\mathbf{D}^s_a} ) \right)
\|\bfz\|_{\mathbf{D}^s_a}
\\
\leqslant
C_1 C_{\varrho}
\|\bfg\|_{\mathbf{V}_{\sigma,D}^{0,s}}
\left(
c_V C_1 C_{\varepsilon} C_d C_{\varrho}
\|\bfu_1\|_{\mathbf{D}^s_a}
\|\theta_2\|_{W^{2,r}(\Omega)} +
c_V C_{\varepsilon} \varrho^{\sharp}
C_d \|\bfu_2\|_{\mathbf{D}^s_a}
\right)
\|\sigma\|_{D^r_{\kappa}}
\\
+ \left[ C_1 C_{\varrho}
\|\bfg\|_{\mathbf{V}_{\sigma,D}^{0,s}}\left(
c_V C_{\varepsilon} \varrho^{\sharp} C_d \|\theta_1\|_{W^{2,r}(\Omega)}
+
\alpha_1 \nu C_e \|\bfu_1\|_{\mathbf{D}^s_a}
+
\alpha_1
\nu
C_e
\|\bfu_2\|_{\mathbf{D}^s_a}
\right)
\right.
\\
\left.
\varrho_0 C_b (\|\bfu_1\|_{\mathbf{D}^s_a}+\|\bfu_2\|_{\mathbf{D}^s_a} )
\right]\|\bfz\|_{\mathbf{D}^s_a}.
\end{multline}
Adding \eqref{ineq_uniq_010} and \eqref{ineq_uniq_011} together
 we get the inequality of the form
\begin{equation*}
\|\sigma\|_{D^r_{\kappa}}
+
\|\bfz\|_{\mathbf{D}^s_a}
\leqslant
R_1(\bfu_1,\bfu_2,\theta_1,\theta_2)
\|\sigma\|_{D^r_{\kappa}}
+
R_2(\bfu_1,\bfu_2,\theta_1,\theta_2)\|\bfz\|_{\mathbf{D}^s_a}.
\end{equation*}
Thus, if $R_1(\bfu_1,\bfu_2,\theta_1,\theta_2)<1$
and $R_2(\bfu_1,\bfu_2,\theta_1,\theta_2)<1$, we have
$\|\sigma\|_{D^r_{\kappa}}=0$ and $\|\bfz\|_{\mathbf{D}^s_a}=0$.
Therefore $\theta_1 = \theta_2$ and $\bfu_1=\bfu_2$.


{
\appendix

\section{Regularity of solutions to appropriate linear elliptic problems}
\label{app:regularity}
In this Appendix, we discuss appropriate linear boundary value problems
for the Poisson equation and the Stokes system in the channel $\Omega$ with
the mixed boundary  conditions.
We establish regularity results for weak solutions
based on the assumptions on the geometry of the channel $\Omega$
and provided
the data of the problems are sufficiently smooth.
Recall that $\Gamma_D$ and $\Gamma_N$ belong to the class
${C}^{\infty}$ and form an angle $\omega_{\mathcal{M}}=\pi/2$ (in the
sense of tangential planes) at all points of $\mathcal{M}$ (the set
in which boundary conditions change their type).

\subsection{The mixed problem for the Stokes system}
We consider the problem
\begin{eqnarray}
-  \Delta \bfu + \nabla P
&=&\bff \qquad \,  \textmd{ in } \Omega,
\label{Stok 20}
\\
\nabla \cdot \bfu &=& 0  \qquad \;\,  \textmd{ in } \Omega,
\label{rovnice kontinuity S0}
\\
\bfu &=& {\bf0}  \qquad \;\, \textmd{ on }  \Gamma_D,
\label{dirichlet S0}
\\
-P\bfn+\frac{\partial \bfu}{\partial \bfn}&=& {\bf0}
\qquad \;\, \textmd{ on } \Gamma_N.
\label{neumann S0}
\end{eqnarray}
Without loss of generality we suppose that
the viscosity of the fluid is normalized to one ($\nu=1$).

For arbitrary real $p \in (1,\infty)$ and
$\bfdelta = (\delta_1,\dots, \delta_d)$, $\delta_i > -2/p$,
$i=1,\dots,d$,
we denote by $\mathcal{W}^{k,p}_{\bfdelta}(\Omega)$
the weighted Sobolev space with the norm
(see e.g. \cite[Chapter 8.3.1.]{MazRoss2010})
\begin{displaymath}
\|\varphi\|_{\mathcal{W}^{k,p}_{\bfdelta}(\Omega)}
=\left(\int_{\Omega} \prod_{i=1}^d r_i(\bfx)^{p \delta_i} \sum_{|\bfalpha|\leq k}
|\partial^{\bfalpha}_{\bfx} \varphi(\bfx)|^p \, {\rm d}\Omega
\right)^{1/p},
\end{displaymath}
$r_i(\bfx)=\textmd{dist}(\bfx,M_i)$.

\medskip

Let $\bfx_0$ be an arbitrary point of $\mathcal{M}$.
Define a dihedral angle $\mathcal{D}$,
$\bfx_0 \in \mathcal{D}$, with faces $\gamma_D$ and $\gamma_N$, such
that $\gamma_D$ and $\gamma_N$ are tangential planes to $\Gamma_D$ and $\Gamma_N$,
respectively, at the point $\bfx_0$ (see Figure \ref{dyhedral_angle}).

\begin{figure}[h]
  \includegraphics[angle=0,width=6cm]{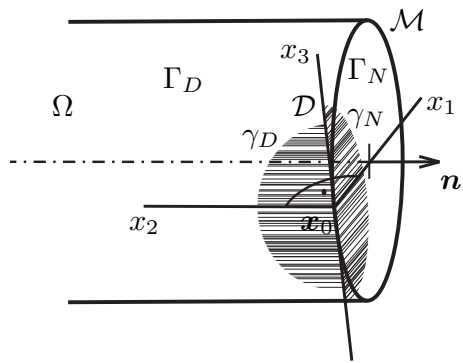}
\caption{Dyhedral angle $\mathcal{D}$.}
\label{dyhedral_angle}
\end{figure}

We choose $\bfx_0$ as the origin and the Cartesian coordinate system in such a way that
$\mathcal{D}$ can be expressed as follows
$$
\mathcal{D} = \left\{[x_1,x_2,x_3]\in \mathbb{R}^3,\; [x_1,x_2] \in K, \; x_3 \in \mathbb{R}  \right\},
$$
where $K=\left\{[x_1,x_2]=[r \cos \omega,r \sin \omega ]; \;
0<r<\infty, \; 0<\omega<\pi/2  \right\}$ is an infinite angle with the sides
$\Gamma_{KD} = \left\{ [x_1,x_2]\in \mathbb{R}^2,\; \omega=\pi/2 \right\}$
and
$\Gamma_{KN} = \left\{ [x_1,x_2]\in \mathbb{R}^2,\; \omega=0 \right\}$.
There is a ball-neighborhood $\mathcal{U}(\bfx_0)$ of the point $\bfx_0$
with radius $r_0$ such that
the domain $\Omega$ is diffeomorphic to the dihedral angle $\mathcal{D}$
in the neighborhood $\mathcal{U}(\bfx_0)$.
Consider a cut-off function  $\chi(|\bfx|)
\in {C}^{\infty}(\mathbb{R}^3)$,
$0 \leq \chi(|\bfx|) \leq 1$,
\begin{displaymath}
\chi(|\bfx|) =  \quad  \left\{
\begin{array}{ccl}
1 & {\rm for} &  |\bfx| < r_0 / 2 , \\
0 & {\rm for} &  |\bfx| > r_0 .
\end{array} \right.
\end{displaymath}
Multiplying the system of equations \eqref{Stok 20}--\eqref{neumann S0}
by $\chi$ we get the boundary value problem for $\tilde{\bfu} = \chi \bfu$
and $\tilde{P} = \chi P$ in the dihedral angle $\mathcal{D}$
\begin{eqnarray}
-  \Delta \tilde{\bfu} + \nabla \tilde{P}
&=&\tilde{\bff} \qquad \,  \textmd{ in } \mathcal{D},
\label{NavStok 2}
\\
\nabla \cdot \tilde{\bfu} &=& \tilde{g}  \qquad \;\,  \textmd{ in } \mathcal{D},
\label{rovnice kontinuity S}
\\
\tilde{\bfu} &=& {\bf0}  \qquad \;\, \textmd{ on }  \gamma_D,
\label{dirichlet S}
\\
-\tilde{P}\bfn +\frac{\partial \tilde{\bfu}}{\partial \bfn}&=& {\bf0}
 \qquad \;\, \textmd{ on }  \gamma_N,
\label{neumann S}
\end{eqnarray}
where
\begin{equation}\label{g}
\tilde{\bff} =-\bfu\Delta \chi
-2\frac{\partial \bfu}{\partial x_1}\frac{\partial \chi}{\partial x_1}
-2 \frac{\partial \bfu}{\partial x_2}\frac{\partial \chi}{\partial x_2}
-2 \frac{\partial \bfu}{\partial x_3}\frac{\partial \chi}{\partial x_3}
+\bff \chi +
P(\nabla \chi)
\end{equation}
and
\begin{equation}
\tilde{g} = \bfu \cdot \nabla\chi.
\nonumber
\end{equation}
After the application of the Fourier transform with respect to $x_3$,
$x_3 \rightarrow \eta$, and letting $\eta=0$ we get the corresponding
two-dimensional problem in the plane angle $K$ with the sides $\Gamma_{KD}$ and $\Gamma_{KN}$
\begin{align}
-  \Delta_{\bar{x}} (\hat{u}_1,\hat{u}_2)  +  \nabla_{\bar{x}} \hat{P}
&=
(\hat{f}_1,\hat{f}_2)  &&\textmd{in}
\quad K,
\label{stationary S 2D}
\\
-  \Delta_{\bar{x}} \hat{u}_3 &= \hat{f}_3 && \textmd{in} \quad K,
\\
\nabla_{\bar{x}} \cdot (\hat{u}_1,\hat{u}_2) &= \hat{g} && \textmd{in} \quad K,
\label{rovnice kontinuity S}
\\
(\hat{u}_1,\hat{u}_2,\hat{u}_3) &= (0,0,0)
&& \textmd{on} \quad \Gamma_{KD},
\label{dirichlet S}
\\
- \hat{P}\bfn_{\bar{x}} +  \frac{\partial (\hat{u}_1,\hat{u}_2)}{\partial \bfn_{\bar{x}}} &= (0,0)
&& \textmd{on} \quad \Gamma_{KN},
\label{neumann S}
\\
\frac{\partial\hat{u}_3}{\partial \bfn_{\bar{x}}} &= 0
&& \textmd{on} \quad \Gamma_{KN}.
\label{neumann P}
\end{align}
Here
$\hat{\bff}=\mathcal{F}_{x_3\rightarrow \eta}[\tilde{\bff}]$,
$\hat{g}=\mathcal{F}_{x_3\rightarrow\eta}[\tilde{g}]$,
$\hat{\bfu}=\mathcal{F}_{x_3\rightarrow \eta}[\tilde{\bfu}]$,
$\hat{P}=\mathcal{F}_{x_3\rightarrow \eta}[\tilde{P}]$.
Under the polar coordinates $(r,\omega)$ the problem
\eqref{stationary S 2D}--\eqref{neumann P} becomes
\begin{eqnarray}
- \left( \frac{\partial ^2 \bar{u}_1}{\partial r^2} + \frac{1}{r}
\frac{\partial \bar{u}_1}{\partial r} +  \frac {1}{r^2} \frac
{\partial ^2 \bar{u}_1}{\partial \omega ^2} \right) + \frac{\partial
\bar{P} }{\partial r}\cos{\omega} -\frac{1}{r}\frac{\partial
\bar{P}}{\partial \omega} \sin \omega
&=&
\bar{f}_1
\quad \textmd{in} \quad \bar{S},
\label{eq20e}
\\
- \left( \frac{\partial ^2 \bar{u}_2}{\partial r^2} + \frac{1}{r}
\frac{\partial \bar{u}_2}{\partial r} +  \frac {1}{r^2} \frac
{\partial ^2 \bar{u}_2}{\partial \omega ^2} \right) + \frac{\partial
\bar{P} }{\partial r}\sin{\omega} + \frac{1}{r}\frac{\partial
\bar{P}}{\partial \omega} \cos \omega
&=&
\bar{f}_2
 \quad \textmd{in} \quad \bar{S},
\label{eq20f}
\\
- \left( \frac{\partial ^2 \bar{u}_3}{\partial r^2} + \frac{1}{r}
\frac{\partial \bar{u}_3}{\partial r} +  \frac {1}{r^2} \frac
{\partial ^2 \bar{u}_3}{\partial \omega ^2} \right)
&=&
\bar{f}_3
 \quad \textmd{in} \quad \bar{S},
\label{eq20g}
\\
\frac{\partial \bar{u}_1}{\partial r}  \cos \omega - \frac{1}{r}
\frac {\partial \bar{u}_1}{\partial \omega} \sin \omega +
\frac{\partial \bar{u}_2}{\partial r}\sin \omega + \frac{1}{r} \frac
{\partial \bar{u}_2}{\partial \omega} \cos \omega
&=&
\bar{g}
 \quad \textmd{in} \quad \bar{S},
\label{eq20h}
\\
\bar{u}_1(r,\pi/2) &=& 0,
\\
\bar{u}_2(r,\pi/2) &=& 0,
\\
\bar{u}_3(r,\pi/2) &=& 0,
\\
\frac{\partial \bar{u}_1}{\partial \omega} (r,0)
&=&
0,
\label{eq25a}
\\
-\bar{P}(r,0)+\frac{\partial\bar{u}_2}{\partial\omega}(r,0)
&=&
0,
\label{eq25b}
\\
\frac{\partial\bar{u}_3}{\partial\omega}(r,0)
&=&
0,
\label{eq25c}
\end{eqnarray}
where $\bar{S}=\left\{ (r,\omega):
0<r<\infty, \, 0<\omega<{\pi / 2} \right\}$
is the infinite plane angle
described in polar coordinates $(r,\omega)$,
$\bar{\bfu}(r,\omega)=\hat{\bfu}(x_1,x_2)$,
$\bar{P}(r,\omega)=\sqrt{x_1^2+x_2^2}\hat{P}(x_1,x_2)$,
$\bar{\bff}(r,\omega)=\hat{\bff}(x_1,x_2)$,
$\bar{g}(r,\omega)=\hat{g}(x_1,x_2)$.

Applying the Mellin transform
$$
\mathcal{M}[\phi]
=
\int_0^{\infty}
r^{-z-1} \phi(r){\rm d}r = \breve{\phi}(z)
$$
we get the
following system of equations depending
on a parameter $z \in \mathbb{C}$  in the interval $(0 , \pi/2)$
\begin{eqnarray}
- \frac {\partial ^2 \breve{u}_1}{\partial \omega ^2} -
z^2\breve{u}_1
+
(z-1)\breve{P}\cos \omega-\frac{\partial \breve{P}}{\partial
\omega} \sin \omega
&=&
\breve{f}_1,
\label{eq22a}
\\
- \frac {\partial ^2\breve{u}_2}{\partial \omega ^2}
-z^2\breve{u}_2+(z-1)\breve{P}\sin\omega
+\frac{\partial\breve{P} }{\partial \omega} \cos \omega
&=&
\breve{f}_2,
\label{eq22b}
\\
- \frac {\partial ^2\breve{u}_3}{\partial \omega ^2}
-z^2\breve{u}_3
&=&
\breve{f}_3,
\label{eq22c}
\\
z\breve{u}_{1}\cos \omega - \frac {\partial
\breve{u}_1}{\partial \omega} \sin \omega
+
z\breve{u}_2\sin \omega
+ \frac{\partial\breve{u}_2 }{\partial\omega} \cos \omega
&=&
\breve{g},
\label{eq22d}
\\
\breve{u}_1(z,\pi/2)
&=&
0,
\label{eq25d}
\\
\breve{u}_2(z,\pi/2)
&=&
0,
\label{eq25e}
\\
\breve{u}_3(z,\pi/2)
&=&
0,
\label{eq25f}
\\
\frac{\partial \breve{u}_1}{\partial \omega} (z,0 )
&=&
0,
\label{eq25a}
\\
-\breve{P}(z,0)+\frac{\partial\breve{u}_2}{\partial\omega}(z,0)
&=&
0,
\label{eq25b}
\\
\frac{\partial \breve{u}_3}{\partial \omega}(z,0 )
&=&
0,
\label{eq25f}
\end{eqnarray}
where
$\breve{f}_1=\mathcal{M}_{r \rightarrow z}[\bar{f}_1]$,
$\breve{f}_2=\mathcal{M}_{r \rightarrow z}[\bar{f}_2]$,
$\breve{f}_3=\mathcal{M}_{r \rightarrow z}[\bar{f}_3]$,
$\breve{g}=\mathcal{M}_{r \rightarrow z}[\bar{g}]$,
$\breve{u}_1=\mathcal{M}_{r \rightarrow z}[\bar{u}_1]$,
$\breve{u}_2=\mathcal{M}_{r \rightarrow z}[\bar{u}_2]$,
$\breve{u}_3=\mathcal{M}_{r \rightarrow z}[\bar{u}_3]$,
$\breve{P}=\mathcal{M}_{r \rightarrow z}[\bar{P}]$.

Now the problem \eqref{eq22a}--\eqref{eq25f} can be treated as
the operator equation
\begin{displaymath}
\mathcal{A}(z)
[\breve{u}_1,\breve{u}_2,\breve{u}_3,\breve{P}]
=
[\breve{f}_1,\breve{f}_2,\breve{f}_3,\breve{g},0,0,0,0,0,0].
\end{displaymath}

Every complex number $z_0$
such that $\ker \mathcal{A}(z_0)\neq {\bf0}$ is said to be an eigenvalue
of $\mathcal{A}(z)$ and the set of all such eigenvalues
is called the spectrum of $\mathcal{A}(z)$.
Note that the problem \eqref{eq22a}--\eqref{eq25f}
with $[\breve{f}_1,\breve{f}_2,\breve{f}_3,\breve{g}]=[0,0,0,0]$ consists of two
decoupled boundary value problems with parameter $z$ which can be handled
separately. The spectrum of $\mathcal{A}(z)$ consists of the numbers
$z_k = 2k+1$, where $k=0,1,2,\dots$,
corresponding to the problem \eqref{eq22c},\eqref{eq25c} and \eqref{eq25f}
with $\breve{f}_3=0$ and the only unknown $\breve{u}_3$
(see \cite[Section 8.3.1]{MazRoss2010}).
In addition, the spectrum of $\mathcal{A}(z)$ includes
the solutions of the transcendental equation
(we refer to \cite[eq. (2.9)]{Benes2009})
\begin{equation}\label{DN_determinant 1}
z^2-
4\cos^2\left(\frac{z\pi}{2}\right)-
\sin^2\left(\frac{z\pi}{2}\right)=0.
\end{equation}
Note that for the roots $z$ of \eqref{DN_determinant 1}
there exists a nontrivial solution of the problem
\eqref{eq22a}--\eqref{eq22b}, \eqref{eq22d}--\eqref{eq25b}, \eqref{eq25d}--\eqref{eq25e}
with
$[\breve{f}_1,\breve{f}_2,\breve{g}]=[0,0,0]$
and the
unknowns $\breve{u}_1$, $\breve{u}_2$ and $\breve{P}$ (see \cite{Benes2009}).

$$$$


Denote by
$\mu_{\mathcal{M}}$ the greatest real number such that the strip
\begin{displaymath}
0 < \Re z  < \mu_{\mathcal{M}}, \; z \in \mathbb{C},
\end{displaymath}
contains only the eigenvalue $z = 1$ of the operator $\mathcal{A}(z)$.

The following result is a consequence of \cite[Theorem 5.5]{MazRoss2007}.
\begin{thm}
[Regularity in weighted Sobolev spaces]
\label{regularity_stokes 1}
Let $\bff \in
 (\mathbf{V}_{\sigma,D}^{1,2})^*$
and
$\bfu\in \mathbf{V}_{\sigma,D}^{1,2}$ be the weak solution
of the problem \eqref{Stok 20}--\eqref{neumann S0}, i.e. satisfying the equation
\begin{equation*}
a(\bfu,\bfv) =(\bff,\bfv)
\end{equation*}
for all $\bfv \in \mathbf{V}_{\sigma,D}^{1,2}$.
Suppose that $\bff \in \mathcal{W}^{0,p}_{\bfdelta}(\Omega)^3$
and the components $\delta_i$ of $\bfdelta$ satisfy the inequalities
$$
\max(0,2-\mu_{\mathcal{M}})<\delta_i + 2/p < 2,
\qquad i=1,\dots,d.
$$
Then $\bfu \in \mathcal{W}^{2,p}_{\bfdelta}(\Omega)^3$ and
\begin{displaymath}
\|\bfu\|_{\mathcal{W}^{2,p}_{\bfdelta}(\Omega)^3}
\leq c(\Omega)
\|\bff\|_{\mathcal{W}^{0,p}_{\bfdelta}(\Omega)^3}.
\end{displaymath}
\end{thm}

\begin{rem}
\label{reg_sup_1}
Let us note (see \cite[Remark 2.2]{Benes2009})
that it can be shown numerically
that there are only two roots of
the equation \eqref{DN_determinant 1},
$z_0\approx 1.352317 $ and $z_{00}=1$, in the strip
$\Re z  \in (0,2)$.
Hence $\mu_{\mathcal{M}} = \Re z_0$ and we set
$s_0=\frac{2}{\Re z_0+2} (\approx 3.087930)$.
Since we consider $4/3 < s < s_0$ (recall \eqref{s_r}), we have
$\max(0,2-\mu_{\mathcal{M}})<2/s<2$.
\end{rem}

For $\bfdelta = {\bf0}$, we obtain the regularity results in nonweighted Sobolev spaces.
The following assertion holds as a consequence of Theorem \ref{regularity_stokes 1}
and Remark \ref{reg_sup_1}.

\begin{cor}
Let
$\bfu\in \mathbf{V}_{\sigma,D}^{1,2}$ be the weak solution
of the problem \eqref{Stok 20}--\eqref{neumann S0} and $\bff \in L^s(\Omega)^3$,
 $4/3 < s < s_0$.
 Then $\bfu \in W^{2,s}(\Omega)^3$ and
\begin{displaymath}
\|\bfu\|_{W^{2,s}(\Omega)^3}
\leq c(\Omega)
\, \|\bff\|_{L^s(\Omega)^3}.
\end{displaymath}
\end{cor}

\subsection{The mixed problem for the Poisson equation}
We consider the problem
\begin{align}
-  \Delta \vartheta
&=
h
&&
\textmd{in}\;\om,
\label{poisson_01}
\\
\vartheta
&=
0
&&\textmd{on}\;\Gamma_{D},
\label{dirichlet_theta_01}
\\
\nabla\vartheta \cdot \bfn &= 0
&&\textmd{on}\;\Gamma_{N}.
\label{neumann_theta_01}
\end{align}

By the Lax-Milgram theorem,
for every $h \in (V_D^{1,2})^*$ there
exists a uniquely determined $\vartheta \in V_D^{1,2}$ (the weak solution to the problem
\eqref{poisson_01}--\eqref{neumann_theta_01}) such
that $\kappa(\vartheta,\varphi) = \langle h ,
\varphi \rangle$ for every $\varphi \in V_D^{1,2}$.
The following regularity result
 is a consequence of \cite[Section 33.5 and Section 26.3]{KufSan}
 (see also \cite[Corollary 8.3.2]{MazRoss2010}).
\begin{prop}
Let $h \in L^p(\Omega)\cap (V_D^{1,2})^*$, $p>1$,
and $\vartheta \in V_D^{1,2}$ be the weak solution to the problem
\eqref{poisson_01}--\eqref{neumann_theta_01}.
Then $\vartheta \in W^{2,p}(\Omega)$ and
\begin{displaymath}
\|\vartheta\|_{W^{2,p}(\Omega)}
\leq c(\Omega)
\, \|h\|_{L^p(\Omega)}.
\end{displaymath}
\end{prop}

}


\subsection*{Acknowledgment}
This research was supported by
the project GA\v{C}R 13-18652S.

\end{document}